\documentclass{llncs}
\usepackage{graphicx}
\usepackage{fancyvrb}
\usepackage{comment}
\usepackage{float,xspace}
\usepackage{multirow,multicol}
\usepackage[linesnumbered,boxed,algoruled,slide,noresetcount]{myalgorithm2e}
\usepackage{setspace}

\newcommand{\START}{{\scriptsize{START}}\xspace}
\newcommand{\DONE}{{\scriptsize{DONE}}\xspace}

\title{A Wait-Free Stack}
\author{Seep Goel, Pooja Aggarwal and Smruti R. Sarangi \\
E-mail: seep.goyal@gmail.com, \{pooja.aggarwal, srsarangi\}@cse.iitd.ac.in
\institute{Indian Institute of Technology, New Delhi, India} }

\begin{document}
\maketitle
\begin{abstract}
In this paper, we describe a novel algorithm to create a concurrent wait-free stack. To the best
of our knowledge, this is the first wait-free algorithm for a general purpose stack. In the past,
researchers have proposed restricted wait-free implementations of stacks, lock-free implementations,
and efficient universal constructions that can support wait-free stacks. The crux of our
wait-free implementation is a fast $pop$ operation that does not modify the stack top; instead,
it walks down the stack till it finds a node that is unmarked. It marks it but does not delete it.
Subsequently, it is lazily deleted by a $cleanup$ operation. 
This operation keeps the size of the stack in check by not allowing the size of the stack to 
increase beyond a factor of $W$ as compared to 
the actual size. All our operations are wait-free and linearizable.
\end{abstract}
\section{Introduction}
\label{sec:intro}
In this paper, we describe an algorithm to create a wait-free stack. A concurrent data structure is said to be
wait-free if each operation is guaranteed to complete within a finite number of steps. In comparison, the data
structure is said to be lock-free if at any point of time, at least one operation is guaranteed to complete
in a finite number of steps. Lock-free programs will not have deadlocks but can have starvation, whereas
wait-free programs are starvation free. Wait-free stacks have not received
a lot of attention in the past, and we are not aware of algorithms that are particularly tailored to creating
a generalized wait-free stack. However, approaches have been proposed to create wait-free stacks with certain
restrictions~\cite{arrayBased}, ~\cite{restrictedStack}, ~\cite{WaitFreeSharedCounter1},~\cite{WaitFreeSharedCounter2}, 
and with universal constructions~\cite{oldUniversal}, ~\cite{newUniversal}. The main reason that it has been 
difficult to create a wait-free stack is because there
is a lot of contention at the stack top between concurrent $push$ and $pop$ operations. It has thus been
hitherto difficult to realize the gains of additional parallelism, and also guarantee completion in a finite
amount of time.

The crux of our algorithm is as follows. We implement a stack as a linked list, where the $top$ pointer
points to the stack top. Each $push$ operation adds an element to the linked list, and updates the $top$ pointer.
Both of these steps are done atomically, and the overall operation is linearizable (appears
to execute instantaneously). However, the $pop$ 
operation does not update the $top$ pointer. This design decision has been made to enable more parallelism,
and reduce the time per operation. It instead scans the list starting from the $top$ pointer till
it reaches an unmarked node. Once, it reaches an unmarked node, it marks it and returns the node as the result
of the $pop$ operation. Over time, more and more nodes get marked in the stack. To garbage collect such
nodes we implement a $cleanup$ operation that can be invoked by both the $push$ and $pop$ operations.
The cleanup operation removes a sequence of $W$ consecutively marked nodes from the list. In our algorithm,
we guarantee that at no point of time the size of the list is more than $W$ times the size of the stack
(number of pushes - pops). This property ensures that $pop$ operations complete within a finite amount of
time. Here, $W$ is a user defined parameter and it needs to be set to an optimal value to ensure the best
possible performance. 

The novel feature of our algorithm is the $cleanup$ operation that always keeps the size of the stack
within limits. It does not allow the number of marked nodes that have already been popped to indefinitely
grow. The other novel feature is that concurrent $pop$ and $push$ operations do not cross each others' paths. 
Moreover, all the $pop$ operations can take place concurrently. This allows us to have a linearizable operation. 
In this paper, we present our basic algorithm along with proofs of important results. Readers can find
the rest of the pseudo code, asymptotic time complexities, and proofs in the appendices. 

\section{Related Work}
\label{sec:rel}
\vspace{-2mm}
In 1986, Treiber~\cite{CopingWithPara} proposed the first lock-free
implementation of a concurrent stack.  He employed a linked list based data
structure, and in his implementation, both the $push$ and $pop$
operations modified the $top$ pointer using CAS instructions.  
Subsequently, Shavit et
al.~\cite{combiningFunnels} and Hendler et al.~\cite{flatCombining} designed a
linearizable concurrent stack using the concept of software combining. Here,
they group concurrent operations, and operate on the entire group.

In 2004, Hendler et al.~\cite{LockFree} proposed a highly scalable lock-free
stack using  an array of lock-free exchangers known as an elimination array.
If a $pop$ operation is paired with a $push$ operation, then the baseline data structure need
not be accessed. This greatly enhances the amount of available parallelism, and
is known to be one of the most efficient implementations of a lock-free stack. 
This technique can be incorporated in our design as well.
Subsequently, Bar-Nissan et al.~\cite{DECS} have augmented this proposal with
software combining based approaches. 
Recently, Dodds et al.~\cite{timestamp} proposed a fast lock-free stack, which uses a
timestamp for ordering the $push$ and $pop$ operations.

The restricted wait-free algorithms for the stack data structure proposed so far by the researchers are summarized
in Table ~\ref{tab:rel}.

\begin{table*}[!htb]
\scriptsize
\begin{center}

\begin{tabular}{|l|l|l| }
\hline
                                                         
{\em Author}			&{\em Primitives}	&{\em Remarks}\\\hline

\hline  
\multirow{2}{*}{Herlihy 1991~\cite{oldUniversal}}		&\multirow{2}{*}{CAS}	&1. Copies every global update to the
private copy of every thread.
\\
	&							&2. Replicates the stack data structure $N$ times ($N \rightarrow$ \# threads).
\\\hline

Afek et al.~\cite{arrayBased}  	&F\&A,	 		&1. Requires a semi-infinite array (impractical).
\\  
2006]				&TAS			&2. Unbounded stack size.
\\\hline 

Hendler et al.~\cite{WaitFreeSharedCounter1}	&\multirow{2}{*}{DCAS}	&1. DCAS not supported in modern hardware 
\\  
2006]						&			& 2. Variation of an implementation of a shared counter. 
\\\hline

Fatourou et al.~\cite{newUniversal}	&LL/SC,	&1. Copies every global update to the private copy of every thread. \\
2011]			&F\&A	&2. Relies on wait-free implementation of F\&A in hardware.
\\\hline 
David et al. 2011 ~\cite{restrictedStack}		&BH Object		&1. Supports at the most two concurrent pop operations
\\\hline      
\hline
\multicolumn{3}{||c||}{CAS $\rightarrow$ compare-and-set, TAS $\rightarrow$ test-and-set, LL/SC $\rightarrow$
load linked-store conditional} \\
\multicolumn{3}{||c||}{DCAS $\rightarrow$ double location CAS, F\&A $\rightarrow$ fetch-and-add, BH Object (custom
object~\cite{restrictedStack})} \\
\hline
\end{tabular}
\caption{ Summary of existing restricted wait-free stack algorithms  \label{tab:rel}}
\end{center}

\vskip -3mm
\end{table*}

The $wait$-$free$ stack proposed in~\cite{arrayBased} employs a semi-infinite
array as its underlying data structure. A $push$ operation obtains a unique index in the
array (using getAndIncrement()) and writes its value to that index. A $pop$
operation starts from the top of the stack, and traverses the stack towards the bottom.
It marks and returns the first unmarked node that we find. Our $pop$ operation
is inspired by this algorithm. Due to its unrestricted stack size, this algorithm
is not practical.

David et al.~\cite{restrictedStack} proposed another class of
 restricted stack implementations.
Their implementation can support a maximum of two concurrent $push$ operations.
Kutten et al.~\cite{WaitFreeSharedCounter1,WaitFreeSharedCounter2} suggest an
approach where a wait-free shared counter can be adapted to
create wait-free stacks. However, their algorithm requires the DCAS (double
CAS) primitive, which is not supported in contemporary hardware. 

Wait-free universal constructions are generic algorithms that can be used to create linearizable implementations of any
object that has valid sequential semantics.  The inherent drawback of these approaches is that they typically have
high time and space overheads (creates local copies of the entire (or partial) data structure).
  A recent proposal by Fatourou et al.~\cite{newUniversal} can be used to implement stacks
and queues.  The approach derives its performance improvement over the widely accepted universal construction of
Herlihy~\cite{oldUniversal} by optimizing on the number of shared memory accesses. 
\section{The Algorithm}

\subsection{Basic Data Structures}
Algorithm~\ref{alg:node} shows the $Node$ class, which represents a node in a stack. 
It has a $value$, and pointers to the next ($nextDone$) and previous nodes ($prev$) respectively.
Note that our stack is not a doubly linked list, the next pointer $nextDone$ is only used for reaching consensus on which node will be added next in the stack. 

 To support $pop$ operations, every node has a $mark$ field. The $pushTid$ field contains 
the id of the thread that created the request. The $index$ field and $counter$ is an atomic integer and is used 
to clean up the stack.
\vspace{-6mm}
\begin{algorithm}
\scriptsize
\SetAlgoLined
\textbf{class Node}{}\\
		\hspace{5mm}$int$ $value$\\
		\hspace{5mm}$AtomicMarkableReference<Node>$ $nextDone$\\
		\hspace{5mm}$AtomicReference<Node>$ $prev$\\
		\hspace{5mm}$AtomicBoolean$ $mark$\\
		\hspace{5mm}$int$ $pushTid$\\
		\hspace{5mm}$AtomicInteger$ $index$\\	
		\hspace{5mm}$AtomicInteger$ $counter$  /* initially set to 0 */\\ 	
\caption{The Node Class} \label{alg:node}
\end{algorithm}
\vspace{-6mm}

In our wait-free stack, the nodes are arranged as a linked list. Initially, the list contains 
only the $sentinel$ node, which is a dummy node. As we push elements, the list starts to grow. 
The $top$ pointer points to the current stack top. 

\normalsize
\vspace{-4mm}
\subsection{High level Overview}

The $push$ operation starts by choosing a phase number (in a monotonically increasing manner), which is greater than
the phase numbers of all the existing push operations in the system. This phase
number along with a reference to the node to be pushed and a flag indicating
the status of the $push$ operation are saved in the $announce$ array in an
atomic step. After this, the thread $t_i$ scans the $announce$ array and finds
out the thread $t_j$, which has a $push$ request with the least phase
number. Note that, the thread $t_j$ found out by $t_i$ in the last step might
be $t_i$ itself. Next, $t_i$ helps $t_j$ in completing $t_j$'s operation. At
this point of time, some threads other than $t_i$ might also be trying to help
$t_j$, and therefore, we must ensure that $t_j$'s operation is applied exactly
once. This is ensured by mandating that for the completion of any $push$
request, the following steps must be performed in the exact specified order:

\begin{enumerate} 
\item Modify the state of the stack in such a manner
that all the other $push$ requests in the system must come to know that a
$push$ request $p_i$ is in progress and additionally they should be able to figure out
the details required by them to help $p_i$.  
\item Update the status flag to
$DONE$ in $p_i$'s entry in the $announce$ array.  
\item Update the $top$
pointer to point to the newly pushed node.  
\end{enumerate}


The $pop$ operation has been designed in such a manner that it does not update
the $top$ pointer. This decision has the dual benefit of eliminating the
contention between concurrent $push$ and $pop$ operations, as well as enabling the
parallel execution of multiple $pop$ operations. The $pop$ operation starts by
scanning the linked list starting from the stack's top till it reaches an
unmarked node. Once, it gets an unmarked node, it marks it and returns the node
as a result of the $pop$ operation.  Note that there is no helping in the case
of a $pop$ operation and therefore, we do not need to worry about a $pop$
operation being executed twice. Over time, more and more nodes get marked in
the stack. To garbage collect such nodes we implement a $clean$ operation
that can be invoked by both the $push$ and $pop$ operations. 

\vspace{-4mm}

\subsection{The Push Operation}
The first step in pushing a node is to create an instance of the $PushOp$ class. It contains the 
reference to a Node ($node$), a Boolean variable $pushed$ that indicates the status of the 
request, and a phase number ($phase$) to indicate the age of the request.
Let us now consider the $push$ method (Line~\ref{line:push}).
 We first get the phase number by atomically incrementing a global counter. 
Once the $PushOp$ is created and its phase is initialized, it is saved in the $announce$ array. 
Subsequently, we call the function $help$ to actually execute the $push$ request. 

The $help$ function (Line~\ref{line:help}) finds the request with the least phase number that has 
not been pushed yet. If there is no such request, then it returns. Otherwise it helps that request 
($minReq$) to complete by calling the $attachNode$ method. 
After helping $minReq$, we check if the request that was helped is the same as the request that was 
passed as an argument to the $help$ function ($request$) in Line~\ref{line:help}. If they are 
different requests, then we call $attachNode$ for the request $request$ in Line~\ref{line:attachreq}. 
This is a standard construction to make a lock-free method wait-free (refer to ~\cite{artOfMulti}).

In the $attachNode$ function, we first read the value of the $top$ pointer, and its $next$ field. 
If these fields have not changed between Lines~\ref{read1} and~\ref{line:notlast}, then we try to find the status of the request in Line~\ref{line:ifdone}. 
Note that we check that $next$ is equal to null, and $mark$ is equal to false in the previous line 
(Line~\ref{line:nextnull}). The $mark$ field is made true after the $top$ pointer has been updated. Hence, in 
Line~\ref{line:nextnull}, if we find it to be true then we need to abort the current iteration and read the 
$top$ pointer again.
\vspace{-6mm}
\begin{algorithm}[!htb]
\scriptsize
\SetAlgoLined

\textbf{class PushOp}{}\\
		\hspace{5mm}$long$ $phase$\\
		\hspace{5mm}$boolean$ $pushed$\\
		\hspace{5mm}$Node$ $node$\\
$AtomicReferenceArray<PushOp>$ $announce$\\
\textbf{push}($tid$, $value$){} \label{line:push}\\
		$phase$ $\leftarrow$ $globalPhase.getAndIncrement()$\\
		$request$ $\leftarrow$ new $PushOp$($phase$,$false$,$new$ $Node$($value$,$tid$)) \\
		announce[tid] $\leftarrow$ $request$\\
		$help$($request$)\\

\caption{The Push Method} \label{alg:PushOp}


\SetAlgoLined
\textbf{help}($request$){} \label{line:help}\\
		($minTid$, $minReq$)  $\leftarrow$ $min_{req.phase}$ \{ ($i$, $req$) $\mid$ $0 \le i < N$, $req = announce[i]$,
!req.pushed  \} \\
		\If{ (minReq $==$ null) $\mid\mid$ ($minReq.phase > request.phase$)} {
			$return$ 
		}
		attachNode($minReq$) \\
		\If{$minReq$ $\ne$ $request$}{
			attachNode($request$) \label{line:attachreq}
		}


\scriptsize
\SetAlgoLined	
\textbf{attachNode}($request$){}\\
		\While{!$request.pushed$}
		{
			$last$ $\leftarrow$ $top.get$()\\
			($next$, $done$) $\leftarrow$ $last.nextDone.get$()\\ \label{read1}
			\If{$last$ == $top.get$() \label{line:notlast}} 
			{
				\If{$next$ == $null$ \&\& $done$ = \textbf{false} \label{line:nextnull}}
				{
					\If{!$request.pushed$  	\label{line:ifdone}}  
					{
						$myNode$ $\leftarrow$ $request.node$\\
						$res$ $\leftarrow$ $last.nextDone.compareAndSet$( ($null$,\textbf{false}), ($myNode$, \textbf{false})) \label{line:pushcas} \\
						
						\If{res}
						 {
							
							$updateTop$()\\
							last.nextDone.compareAndSet ( ($myNode$, \textbf{false}), (null,\textbf{true})) \\
							$return$
						}
					}
				}
					$updateTop$() \label{line:otherupdatetop}
			}
		}

\end{algorithm}
\vspace{-8mm}

After, we read the status of the request, and find that it has not completed, we proceed to update the $next$ 
field of the stack top in Line~\ref{line:pushcas} using a compare-And-Set (CAS) instruction. The aim is to change the pointer in 
the $next$ field from $null$ to the node the push request needs to add. 
If we are successful, then we need to update the $top$ pointer by calling the 
function, $updateTop$. After the $top$ pointer has been updated, we do not really need the $next$ field for 
subsequent $push$ requests. It will not be used. However, concurrent requests need to see that $last.nextDone$ 
has been updated. The additional compulsion to delete the contents of the pointer in the $next$ field is that 
it is possible to have references to deleted nodes via the $next$ field. The garbage collector in this case will 
not be able to remove the deleted nodes. Thus, after updating the top pointer, we set the $next$ field's 
pointer to $null$, and set the $mark$ to true. If a  concurrent request reads the mark to be true, then it can be 
sure, that the $top$ pointer has been updated, and it needs to read it again.

If the CAS instruction fails, then it means that another concurrent request has successfully performed a CAS 
operation. However, it might not have updated the $top$ pointer. It is thus necessary to call the $updateTop$ 
function to help the request complete. 

\vspace{-8mm}
\begin{algorithm}
\scriptsize
\SetAlgoLined	
\textbf{updateTop}(){}\\
		$last$ $\leftarrow$ $top.get$()\\
		($next$, $mark$) $\leftarrow$ $last.nextDone.get$()\\
		\If{ $next \ne null$  \label{line:nextnonnull} }
		{
	    	$request$ $\leftarrow$ $announce.get$($next$.$pushTid$) \\
			\If{$last==top.get$() \&\& $request.node==next$}
			{
				/* Add the request to the stack and update the top pointer */ \\
				$next.prev.compareAndSet$($null$, $last$) \label{line:push_prev} \\
				$next.index$ $\leftarrow$ $last.index$ +1 \label{line:index}\\
				$request.pushed$ $\leftarrow$ \textbf{true} \label{line:announce} \\
				$stat$ $\leftarrow$ top.compareAndSet($last$, $next$) \label{line:updatetop} \\
				/* Check if any cleaning up has to be done */ \\
				\If{next.index \% $W$ == 0 \&\& $stat$ == \textbf{true} \label{line:onethread}} 
				{
					tryCleanUp(next) \label{line:trycleanup}
				}
			}
		}
\caption{The $updateTop$ method} \label{alg:updatetop} 
\end{algorithm}

\vspace{-8mm}
\normalsize
The $updateTop$ method is shown in Algorithm~\ref{alg:updatetop}. We read the $top$ pointer, and the $next$ 
pointer. If $next$ is non-null, then the request has not fully completed. It is necessary to help it complete. 
After having checked the value of the $top$ pointer, and the value of the $next$ field, we proceed to connect 
the newly attached node to the stack by updating its $prev$ pointer. We set the value of its $prev$ pointer in 
Line~\ref{line:push_prev}. Every node in the stack has an index that is assigned in a monotonically increasing 
order. Hence, in Line~\ref{line:index}, we set the index of $next$ to 1 plus the index of $last$.
Next, we set the $pushed$ field 
of the $request$ equal to true. The point of
linearizability is Line~\ref{line:updatetop}, 
where we update the $top$ pointer to point to $next$ instead of
$last$. This completes the $push$ operation. 

We have a cleanup mechanism that is invoked once the index of a node becomes a multiple of a constant, $W$. 
We invoke the $tryCleanUp$ method in Line~\ref{line:trycleanup}. It is necessary that the $tryCleanUp()$ method
be called by only one thread.  Hence, the thread that successfully performed a CAS on the top pointer
calls the $tryCleanUp$ method if the index is a multiple of $W$. 

\vspace{-3mm} 
\subsection{The Pop Operation}
Algorithm~\ref{alg:pop} shows the code for the $pop$ method. We read the value of 
the $top$ pointer and save it in the local variable, $myTop$. This is the only instance
in this function, where we read the value of the $top$ pointer. Then, we walk back
towards the sentinel node by following the $prev$ pointers (Lines~\ref{line:popwhile} --
\ref{line:popwhileend}). We stop when we are successfully able to set the mark of a node
that is unmarked. This node is logically ``popped'' at this instant of time. 
If we are not able to find any such node, and we reach the sentinel node, then we throw
an $EmptyStackException$. 

\vspace{-6mm}
\begin{algorithm}[!htb]
\scriptsize
\SetAlgoLined

\textbf{pop}(){}\\
	$mytop$  $\leftarrow$  $top.get$()\\ \label{alg:prev_top}
	$curr$  $\leftarrow$  $mytop$\\
	\label{alg:while_mark}
	\While{ $curr \ne sentinel$ }  {  \label{line:popwhile} 
			$mark$  $\leftarrow$ $curr.mark.getAndSet$($true$)\\  
			\If{$!mark$}
			{
				$break$
			}
			$curr$  $\leftarrow$ $curr.prev$
	} \label{line:popwhileend}

	\If{$curr == sentinel$}
	{
		/* Reached the end of the stack */ \\
			$throw$ $new$ $EmptyStackException$() 
	}
	/* Try to clean up parts of the stack */ \\
	tryCleanUp(curr) \label{line:popclean}

	return $curr$
\caption{The Pop Method} \label{alg:pop}
\end{algorithm}
\vspace{-6mm}

After logically marking a node as popped, it is time to physically delete it. We thus
call the $tryCleanUp$ method in Line~\ref{line:popclean}. The $pop$ method returns the
node that it had successfully marked.
\vspace{-3mm}
\subsection{The CleanUp Operation}

The aim of the $clean$ method is to clean a set of $W$ contiguous entries in the list (indexed by 
the $prev$ pointers). Let us start by defining some terminology. Let us define a range of $W$ 
contiguous entries, which has four distinguished nodes as shown in Figure~\ref{fig:rangew}.

A range starts with a node termed the $base$, whose index is a multiple of $W$. Let us now define
$target$ as $base.prev$. The node at the end of a range is $leftNode$. Its index is equal to 
$base.index + W - 1$. Let us now define a node $rightNode$ such that $rightNode.prev = leftNode$. 
Note that for a given range, the $base$ and $leftNode$ nodes are fixed, whereas the $target$ and
$rightNode$ nodes keep changing. $rightNode$ is the base of another range, and its index is
a multiple of $W$. 

The $push$ and the $pop$ methods call the function $tryCleanUp$. The $push$ method calls it when
it pushes a node whose index is a multiple of $W$. This is a valid $rightNode$. It walks back 
and increments the counter of the $base$ node of the previous range. We ensure that only one thread
(out of all the helpers) does this in Line~\ref{line:onethread}. Similarly, in the $pop$ function,
whenever we mark a node, we call the $tryCleanUp$ function. Since the $pop$ function does not
have any helpers, only one thread per node calls the $tryCleanUp$ function. Now, inside
the $tryCleanUp$ function, we increment the counter of the $base$ node. Once, a thread increments it to 
$W+1$, it invokes the $clean$ function. Since only one thread will increment the counter to $W+1$, 
only one thread will invoke the $clean$ function for a range. 
\vspace{-4mm}
\begin{figure}[!htb]
\begin{center}
\includegraphics[width=0.75\columnwidth]{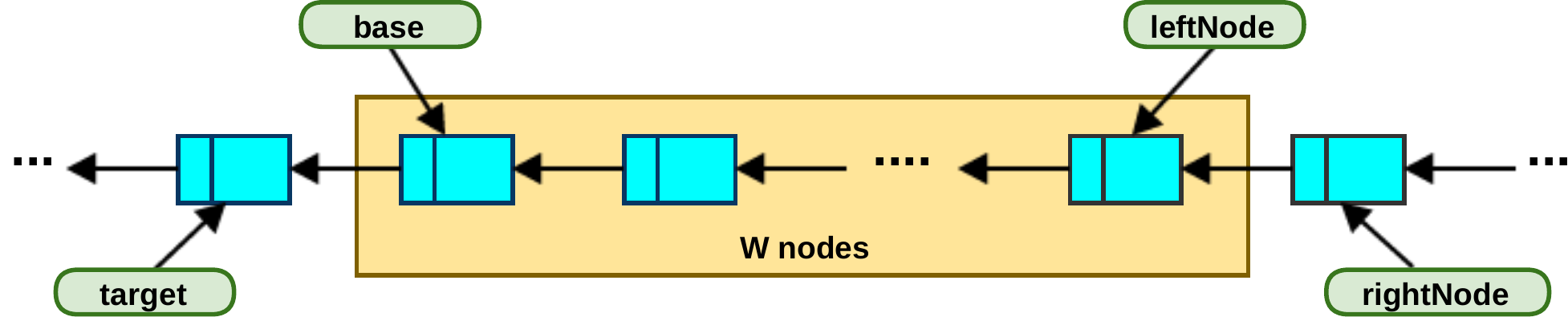}
\caption{A range of $W$ entries \label{fig:rangew} }
\end{center}
\end{figure}
\vspace{-8mm}

\vspace{-7mm}
\begin{algorithm}
\scriptsize
\SetAlgoLined
\textbf{tryCleanUp(myNode)} \\
	$temp$ $\leftarrow$ $myNode.prev$\\
	\While{$temp$ $\ne$ $sentinel$}
	{
		\If{$temp.index$() $\%$ W == 0}
		{
			\If{$temp.counter.incrementAndGet$ == W + 1}
			{
				$clean$(getTid(), $temp$) \label{line:callclean}
			}
			$break$\\
		}
		$temp$ $\leftarrow$ $temp.prev()$\\
	}
\caption{The $tryCleanUp$ method} \label{alg:tryClean}
\end{algorithm}
\vspace{-7mm}
The functionality of the $clean$ function is very similar to the $push$ function. 
Here, we first create a $DeleteRequest$ that has four fields: $phase$ (similar to phase in $PushOp$), 
$threadId$, $pending$ (whether the delete has been finished or not), and  the value of the $base$ node. 
Akin to the $push$ function, we add the newly created $DeleteRequest$ to a global array of $DeleteRequest$s.
Subsequently, we find the pending request with the minimum phase in the array $allDeleteRequests$.

Note that at this stage it is possible for multiple threads to read the same value of the request with the 
minimum phase number. It is also possible for different sets of threads to have found different requests
to have the minimum phase. For example, if a request with phase 2 ($R_2$) got added to the array before the 
request with phase 1 ($R_1$), then a set of threads might be trying to complete $R_2$, and another set might 
be trying to complete $R_1$. To ensure that our stack remains in a consistent state, we want that only one 
set goes through to the next stage.

To achieve this, we adopt a strategy similar to the one adopted in the function $attachNode$. 
Interested readers can refer to the appendices for a detailed explanation of how this is done.
Beyond this point, all the threads will be working on the same $DeleteRequest$ which we term as $uniqueRequest$. 
They will then move on to call the $helpFinishDelete$ function that will actually finish the delete request.

Let us describe the $helpFinishDelete$ function in Algorithm~\ref{alg:helpfinishdelete}. 
We first read the current request from the atomic variable, $uniqueRequest$ in Line~\ref{line:currread}. If
the request is not pending, then some other helper has completed the request, and we can return from the function.
However, if this is not the case, then we need to complete the delete operation. Our aim now is to find the $target$,
$leftNode$, and $rightNode$. We search for these nodes starting from the stack top.

The index of the $leftNode$ is equal to the index of the node in the current request ($currRequest$) + $W-1$. 
$endIdx$ is set to this value in Line~\ref{line:endidx}.
Subsequently, in Lines~\ref{line:right}--\ref{line:sentinel}, we start from the top of the stack, and keep traversing the $prev$
pointers till the index of $leftNode$ is equal to $endIdx$. Once, the equality condition is satisfied,
Lines~\ref{line:right} and \ref{line:left} give us the pointers to the $rightNode$ and $leftNode$ respectively. If
we are not able to find the $leftNode$, then it means that another helper has successfully deleted the nodes. We can
thus return. \\

\vspace{-8mm}
\begin{algorithm}
\scriptsize
\SetAlgoLined
\textbf{helpFinishDelete}(){}\\
		$currRequest$ $\leftarrow$ $uniqueRequest.get$() \label{line:currread} \\
		\If{$!currRequest.pending$}
		{
			$return$
		}
			$endIdx$ $\leftarrow$ $currRequest.node.index + W - 1$	 \label{line:endidx}\\

			$rightNode$ $\leftarrow$ $top.get()$                 /* Search for the request from the $top$ */ \\
		    $leftNode$ $\leftarrow$ $rightNode.prev$ \\	
			\While{$leftNode.index \ne endIdx$ $\&\&$ $leftNode\ne sentinel$} {
				$rightNode$ $\leftarrow$ $leftNode$ \label{line:right} \\
				$leftNode$ $\leftarrow$ $leftNode.prev$ \label{line:left}
			}
			\If{$leftNode = sentinel$} {
				$return$ /* some other thread deleted the nodes */
			} \label{line:sentinel}
		/* Find the target node */ \\
		$target$ $\leftarrow$ $leftNode$  \label{line:targetstart}\\
		\For{i=0; i $< W$; i++}{
			$target$ $\leftarrow$ $target.prev$	
		} \label{line:targetend}
		$rightNode.prev.compareAndSet(leftNode, target)$ /* Perform the CAS operation and delete the nodes */ \label{line:rightcas} \\
		$currRequest.pending$ $\leftarrow$ \textbf{false} /* Set the status of the delete request to not pending*/ \label{line:setpending}\\

\caption{The $helpFinishDelete$ method} \label{alg:helpfinishdelete}
\end{algorithm}
\vspace{-8mm}

The next task is to find the $target$. The $target$ is $W$ hops away from the $leftNode$.
Lines~\ref{line:targetstart}--\ref{line:targetend} run a loop $W$ times to find the target. Note that we shall never
have any issues with null pointers because $sentinel.prev$ is set to $sentinel$ itself. Once, we have found the
target, we need to perform a CAS operation on the $prev$ pointer of the $rightNode$. We accomplish this in
Line~\ref{line:rightcas}. If the $prev$ pointer of $rightNode$ is equal to $leftNode$, then we set it to $target$.
This operation removes $W$ entries (from $leftNode$ to $base$) from the list. The last step is to set the 
status of the $pending$ field in the current request ($currRequest$) to false (see Line~\ref{line:setpending}).

\section{Proof of Correctness}
\vspace{-3mm}
The most popular correctness criteria for a concurrent shared object is {\em
linearizability}~\cite{linearizability}. Linearizability ensures that within the execution interval of every operation
there is a point, called the linearization point, where the operation seems to
take effect instantaneously and the effect of all the operations on the object
is consistent with the object's sequential specification. 
By the property of compositional linearizability, if each method of an object is linearizable we can 
conclude that the complete object is linearizable.
Thus, if we identify the point of linearization for both the push and the pop method in our implementation, 
we can say that our implementation is linearizable and thus establish its correctness.

Interested readers can refer to the appendices, where we show that our 
implementation is legal and push and pop operations complete in a bounded number of steps.

\begin{theorem}
The $push$ and $pop$ operations are linearizable.
\end{theorem}

\begin{proof}
Let us start out by defining the notion of ``pass points''. The pass point of a $push$ 
operation is when it successfully updates the $top$ pointer in the function $updateTop$
(Line~\ref{line:updatetop}). The pass point of the $pop$ operation, is when it successfully marks a node,
or when it throws the {\em EmtpyStackException}.
Let us now try to prove by mathematical
induction on the number of requests
that it is always possible to construct a linearizable execution that is equivalent to
a given execution. In a linearizable execution all the operations are arranged in a sequential
order, and if request $r_i$ precedes $r_j$ in the original execution, then $r_i$ precedes $r_j$
in the linearizable execution as well.

\noindent \textbf{Base Case:} Let us consider an execution with only one pass point. Since
the execution is complete, we can conclude that there was only one request in the system.
An equivalent linearizable execution will have a single request. The outcome of the request
will be an {\em EmptyStackException} if it is a $pop$ request, otherwise it will push a node
to the stack. Our algorithm will do exactly the same in the $pop$ and $attachNode$ methods
respectively. Hence, the executions are equivalent. \\
\noindent \textbf{Induction Hypothesis:} Let us assume that all executions with $n$ requests
are equivalent to linearizable executions. \\
\noindent \textbf{Inductive Step:} Let us now prove our hypothesis for executions with $n+1$
requests. Let us arrange all the requests in an ascending order of the execution times 
of their pass points. Let us consider the last ($(n+1)^{th}$) request just after the pass point of the
$n^{th}$ request. Let the last request be a $push$. If the $n^{th}$ request is also a $push$, 
then the last request will use the $top$ pointer updated
by the $n^{th}$ request. Additionally, in this case 
the $n^{th}$ request will not see any changes made by the
last request. It will update $last.next$ and the $top$ pointer, before the last request updates
them. In a similar manner we can prove that no prior $push$ request will see the last request.
Let us now consider a prior $pop$ request. A $pop$ request scans all the nodes between the
$top$ pointer and the sentinel. None of the pop requests will see the updated $top$ pointer
by the last request because their pass points are before this event. Thus, they have no
way of knowing about the existence of the last request. Since the execution of the first $n$
requests is linearizable, an execution with the $(n+1)^{th}$ push request is also linearizable
because it takes effect at the end (and will appear last in the equivalent sequential order). 

Let us now consider the last request to be a $pop$ operation. A $pop$ operation writes
to any shared location only after its pass point. Before its pass point, it does not do
any writes, and thus all other requests are oblivious of it. 
Thus, we can remove the last request, and the responses of the first $n$ requests will remain the same.
Let us now consider an execution fragment consisting of the first $n$ requests. It is equivalent to
a linearizable execution, $\mathcal{E}$. This execution is independent of the $(n+1)^{th}$ request. 

Now, let us try to create a linearizable execution, $\mathcal{E'}$, which has an event corresponding
to the last request. Since the linearizable execution is sequential, let us represent the request
and response of the last $pop$ operation by a single event, $R$. 
Let us try to modify $\mathcal{E}$ to create $\mathcal{E'}$. Let the sequential execution corresponding
to $\mathcal{E}$ be $\mathcal{S}$.

Now, it is possible that $R$ could have read the $top$ pointer long ago, and
is somewhere in the middle of the stack. In this case, we cannot assume that $R$ is the last
request to execute in the equivalent linearizable execution. Let the state of the stack before the 
$pop$ reads the top pointer be $\mathcal{S'}$. The state $\mathcal{S'}$ is independent of the $pop$ request.
Also note that, all the operations that have arrived after the $pop$ operation have read the $top$ pointer, and 
overlap with the $pop$ operation. The basic rule of $linearizability$ states that, if any operation $R_i$ 
precedes $R_j$ then $R_i$ should precede $R_j$ in the equivalent sequential execution also. Whereas, in case 
the two operations overlap with each other, then their relative order is undefined and any ordering of these 
operations is a valid ordering~\cite{artOfMulti}. 

In this case, we have two possibilities: (I) $R$ returns the node that it had read as the top pointer as an output
of its $pop$ operation, or (II) it returns some other node. 

\noindent \underline{Case I:} In this case, we can consider the point at which $R$ reads the top pointer as the point
at which it is {\em linearized}. $R$ in this case reads the stack top, and pops it.

\noindent \underline{Case II:} In this case, some other request, which is concurrent must have popped the node that
$R$ read as the top pointer. Let $R$ return node $N_i$ as its return value. This node must be between the top pointer
that it had read (node $N_{top}$), and the beginning of the stack. Moreover, while traversing the stack from $N_{top}$
to $N_i$, $R$ must have found all the nodes in the way to be marked. At the end it must have found $N_i$ to be unmarked,
or would have found $N_i$ to be the end of the stack (returns exception).

Let us consider the journey for $R$ from $N_{top}$ to $N_i$. Let $N_j$ be the last node before $N_i$ that has been
marked by a concurrent request, $R_j$. We claim that if $R$ is linearized right after $R_j$, and the rest of the
sequences of events in $\mathcal{E}$ remain the same, we have a linearizable execution ($\mathcal{E'}$). 

Let us consider request $R_j$ and its position in the sequential execution, $\mathcal{S}$. At its point of
linearization, it reads the top of the stack and returns it (according to $\mathcal{S}$). This node $N_j$
is the successor of $N_i$. At that point $N_i$ becomes the top of the stack. At this position, if we insert
$R$ into $S$, then it will read and return $N_i$ as the stack top, which is the correct value. Subsequently,
we can insert the remaining events in $S$ into the sequential execution. They will still return the same set
of values because they are unaffected by $R$ as proved before. 

This proof can be trivially extended to take cleanup operations into account. 
\end{proof}

\section{Conclusion}
The crux of our algorithm is the $clean$ routine, which ensures that the size of the 
stack never grows beyond a predefined factor, $W$. This feature allows for
a very fast $pop$ operation, where we need to find the first entry from the top of the
stack that is not marked. This optimization also allows for an increased amount of parallelism, and
also decreases write-contention on the $top$ pointer because it is not updated by $pop$ operations. 
As a result, the time 
per $pop$ operation is very low. The $push$ operation is also designed to be very fast. It simply
needs to update the $top$ pointer to point to the new data. To provide wait-free guarantees
it was necessary to design a $clean$ function that is slow.
Fortunately, it is not invoked for an average of $W-1$ out of $W$ invocations of $push$ and $pop$. 
We can tune the frequency of the $clean$ operation by varying
the parameter, $W$ (to be decided on the basis of the workload).

\vspace{-2mm}
\bibliographystyle{splncs03}

\newpage
\section* {\textbf{Appendices of A Wait-Free Stack}}
\begin{appendix}

\section{Asymptotic Worst-Case Time Complexity }
Let us now consider the asymptotic worst-case time complexity of the $push$, $pop$ and $clean$ methods
in terms of the number of concurrent threads in the system ($N$), the actual size of the
stack($S$) and the parameter $W$. 

\subsection{The $clean$ method}
The time complexity of the $clean$ method is the same as that of the
$helpDelete$ function. The $helpDelete$ function finds the $delete$ 
request with the minimum phase number, which requires $O(N)$ steps. 
After having found the request with the minimum phase number, it calls 
the $uniqueDelete$ function. The $uniqueDelete$ function contains a $while$ 
loop. In the worst case this $while$ loop might execute $N$ times. Now, when 
we look into the body of this $while$ loop, everything except the call to 
the $helpFinishDelete$ function has O(1) time complexity. The $helpFinishDelete$ 
function contains two loops. The first is a $while$ loop, which traverses 
the stack from the top to the point it finds the desired node. In the 
worst possible case, this loop might end up traversing the complete stack.
We do not allow the size of the stack to increase by more than 
a factor of $W$ as compared to $S$, the worst case time complexity of this loop is therefore
O($WS$). 
The other loop in the function is a $for$ loop, with time complexity O(W). 
So, the worst case time complexity of the $helpFinishDelete$ function is O($WS$) 
and therefore, the worst case time complexity of the clean function is O($NWS$).
The high time complexity of this method is an achilles heel of our algorithm; hence, we
are working on reducing its complexity as well as practical run time.
However, it should be noted that this function is meant to be called infrequently (1 in
$W$ times).

\subsection{The $pop$ method}
In the $pop$ method, everything except the $while$ loop and the call to the 
$tryCleanUp$ function take O(1) time. The $while$ loop is iterated over till 
the time an unmarked mode is encountered. In our algorithm, as soon as $W$ 
consecutive nodes get marked, we issue a $cleanUp$ request for it and at 
any point of time there can be at most $N-1$ $clean$ requests in the system. 
Thus after having traversed at most $WN$ nodes, a $pop$ request is assured to 
find an unmarked node. Now, if we analyze the $tryCleanUp$ method, 
in the worst case scenario, the $while$ loop inside the function will be iterated over 
$W$ times but the $clean$ method will only be called at most once. In fact, the $clean$ 
method is called only once for a group of $W+1$ operations, and therefore, the worst case 
time complexity of the $tryCleanUp$ function, which is O($NWS$), will be incurred very infrequently
(1 in $W$). 
Nevertheless the worst case time complexity of the $pop$ operation is O($NWS$), and the
amortized time complexity (across $W$ pop operations is $O(NS)$. 

\subsection{The $push$ method \label{push_time_complexity}}
The time complexity of the $push$ method is the same as that of the $help$ function. 
Since the $help$ function is supposed to find the request with the least phase number, 
it takes at least O($N$) time. After having found the request with the minimum 
phase number, the $help$ function calls the $attachNode$ function. Note that for any 
$push$ request, the maximum number of times the $while$ loop in the $attachNode$ 
function could possibly execute is of O($N$). Also note that everything inside the $while$ loop, 
except the call to the $updateTop$ function requires only a constant amount of 
time for execution. If the index of the newly pushed node is not a multiple of 
$W$, the $updateTop$'s time complexity is O(1), and therefore the time complexity
of the $push$ operation is O($N$), but if this is not the case, the time complexity 
of the $updateTop$ function becomes dependent on the time 
complexity of the $tryCleanUp$ function, which is O($NWS$) in the worst case. 

All our methods: $push$, $pop$ and $clean$ are bounded wait-free.

\section{Background}
\label{sec:back}

A $stack$ is a data structure that provides $push$ and $pop$ operations with $LIFO$((Last-in-First-Out) semantics.  A
data structure is said to respect LIFO semantics, if the last element inserted is the first to be removed.

\subsection{Correctness}
\label{subsec:linearizability}
The most popular correctness criteria for a concurrent shared object is {\em
linearizability}~\cite{linearizability}. Let us define it formally.

Let us define two kinds of events in a
method call namely {\em invocations (inv)} and {\em responses (resp)}. A
chronological sequence of {\em invocations (inv)} and {\em responses (resp)}
events in the entire execution is known as a {\em history}.  Let a matching
invocation-response pair with a sequence number $i$ be referred to as request
$r_i$. Note that in our system, every invocation has exactly one matching
response.  A request $r_i$ precedes request $r_j$, if
$r_i$'s response comes before $r_j$'s invocation.  This is denoted by $r_i$ $\prec$
$r_j$. A history, $H$, is said to be sequential if an invocation is immediately
followed by its response. A subhistory $(H|T)$ is the subsequence of $H$
containing all the events of thread $T$. Two histories, $H$ and $H'$, are
equivalent if for every thread , $H|T$ $=$ $H'|T$. 
A complete history - {\em
complete(H)} is a history that does not have any pending invocations. 
The
{\em sequential specification} of an object constitutes of the set of all sequential
histories that are correct. A sequential history is {\em legal} if for every object
x, $H|x$ is in the sequential specification of x.

A history $H$ is linearizable if {\em complete(H)} is equivalent to a legal
sequential history, $S$. Additionally, if $r_i$ $\prec$ $r_j$ in
{\em complete(H)}, then $r_i$ $\prec$ $r_j$ in S also. Alternatively we can say,
linearizability ensures that within the execution interval of every operation
there is a point, called the linearization point, where the operation seems to
take effect instantaneously and the effect of all the operations on the object
is consistent with the object's sequential specification. 

\subsection{Progress}
\label{subsec:progress}

Generally, there are two kinds of implementations for a concurrent object:
$blocking$ and $non$-$blocking$. Blocking algorithm use locks. Approaches
that protect critical sections with locks unnecessarily limit parallelism
and are known to be inefficient.

In comparison non-blocking implementations can prove to be much faster.
Such algorithms rely on atomic primitives such as 
compare-And-Set(CAS), LL/SC, and getAndIncrement. They do not have
critical sections.
In this context,  lock-freedom is defined as a property that ensures
that at any point of time at least one thread makes progress.
Or in other words, the system as a whole is always making progress. They
can still have problems of starvation.

Wait-free algorithms provide starvation freedom in addition
to being lock-free. They ensure that every process completes
its operation in a finite number of steps. The wait-free algorithms
have a notion of inherent {\em fairness}, where {\em fairness} measures the degree of 
imbalance across different threads. We quantify fairness as the ratio of the average number
of operations completed by an thread divided by the number of operations completed by the fastest
thread. Figure~\ref{fig:fairness} shows a comparison of the fairness
of our wait-free stack $WF$ with the lock-free stack $LF$ in~\cite{LockFree} 
and the locked stack in~\cite{michaelScottLocked} $LCK$. 
For the $WF$ version, the average {\em fairness} is around 80\%, whereas for $LF$ the {\em fairness} 
goes as low as 50\%, and for $LCK$, it even drops to 25\%.
Also, as shown in figure ~\ref{fig:pdfLF} and ~\ref{fig:pdfWF},
in the case of $WF$, almost all the threads have completed more than 90\% of their work,
whereas for $LF$ only 11 out of 64 threads have completed more than 90\% of their work
and for some threads the percentage of work done is as low as 40\% only.

\begin{figure*}[!htb]
\begin{center}
\begin{tabular}{cc}

\begin{minipage}{.33\textwidth}
\includegraphics[width=0.99\textwidth]{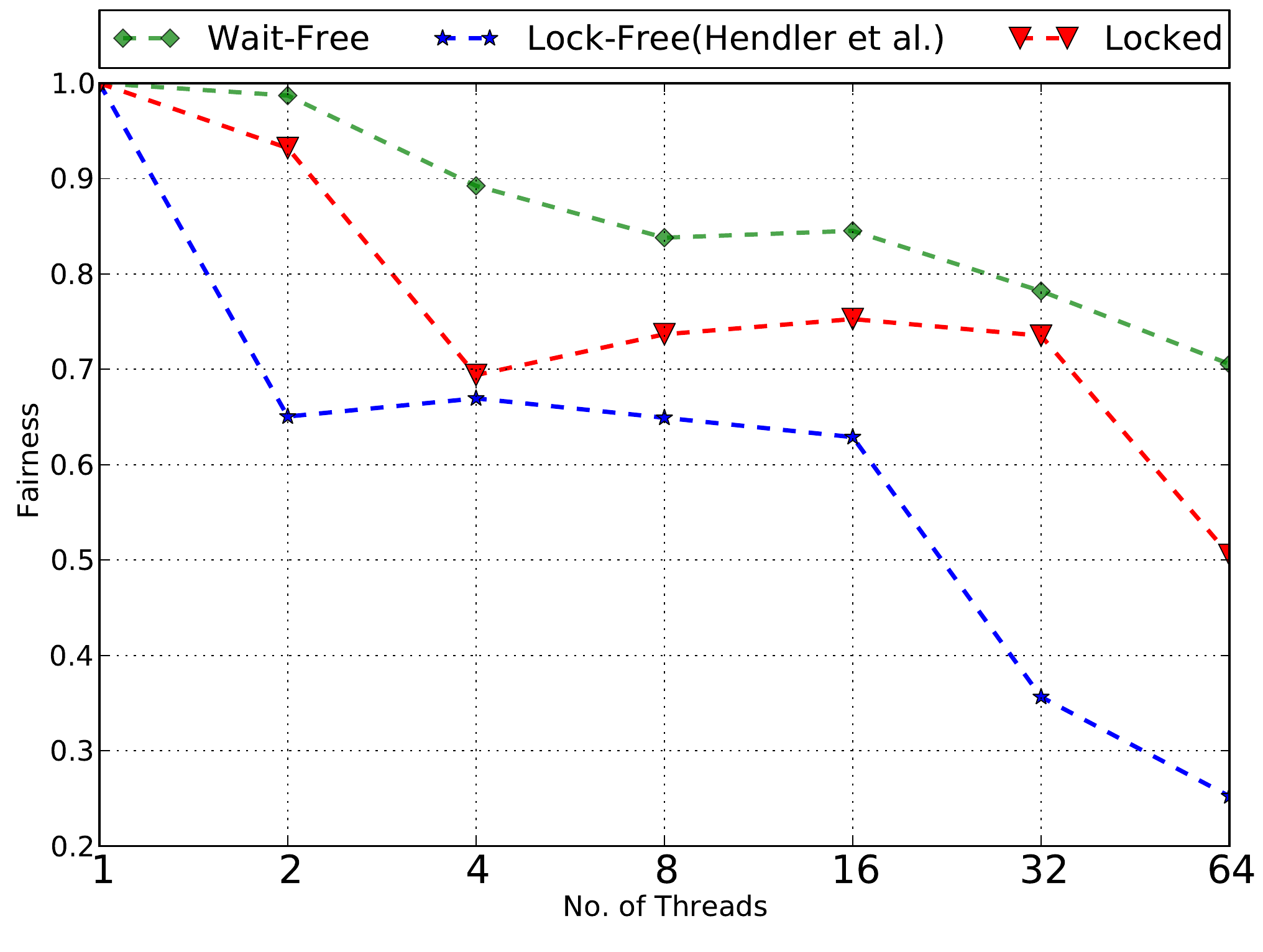}
\caption {Fairness \label{fig:fairness} }
\end{minipage}

\begin{minipage}{.33\textwidth}
\includegraphics[width=0.99\textwidth]{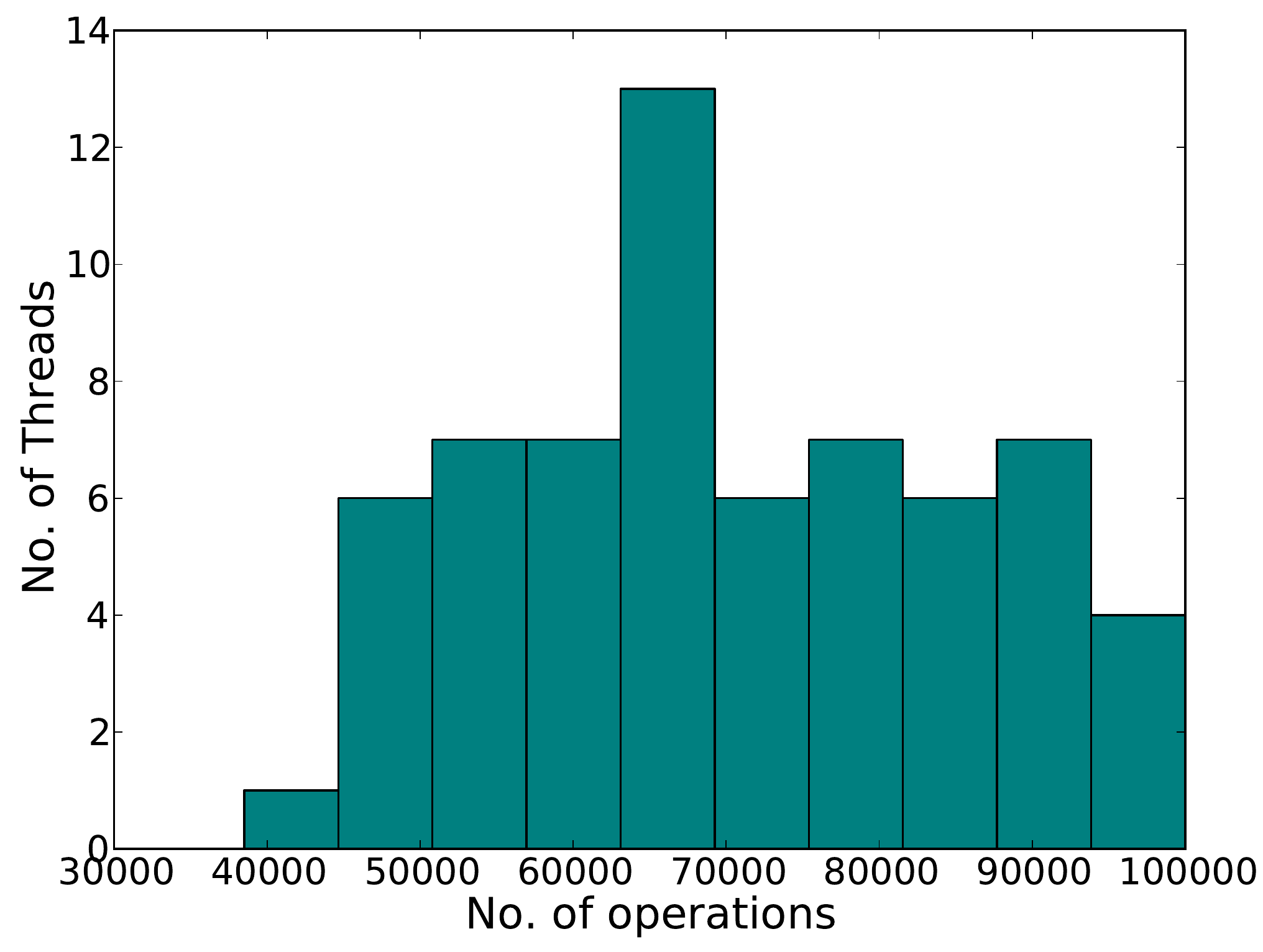}
\caption {Operations done by various threads : Lock-free \label{fig:pdfLF} }
\end{minipage}

\begin{minipage}{.33\textwidth}
\includegraphics[width=0.99\textwidth]{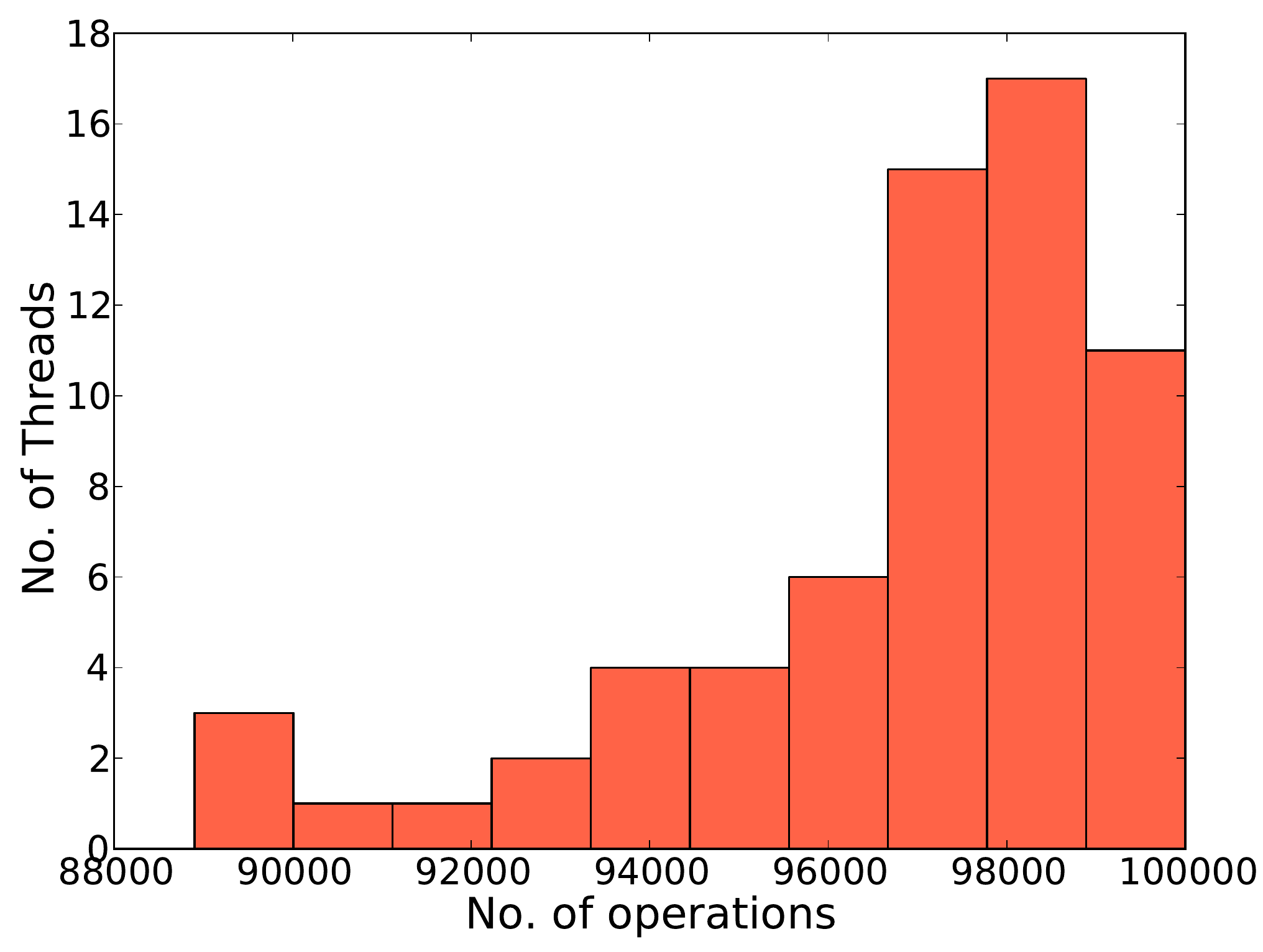}
\caption {Operations done by various threads : Wait-free\label{fig:pdfWF} }
\end{minipage}

\end{tabular}
\end{center}
\end{figure*}

\section{Proof of Correctness}

\begin{lemma}
Every $push$ request is inserted at Line~\ref{line:pushcas} at most once.
\label{lemm:atmost}
\end{lemma}

\begin{proof}
To the contrary, let us assume that the same request is inserted at least twice in Line~\ref{line:pushcas}.
Let us consider the sequence of steps that need to take place. 
\begin{description}
\item [Step 0:] Read the value of $last$ and $next$, and observe that $next = null$. 
\item [Step 1:] We read the status as \START in Line~\ref{line:ifdone}.
\item [Step 2:] The CAS succeeds in Line~\ref{line:pushcas}.
\item [Step 3:] Some thread reads {\em next $\ne$ null} in Line~\ref{line:nextnonnull} ($updateTop$ function).
\item [Step 4:] The status of the node is updated to \DONE by some thread in Line~\ref{line:announce}.
\item [Step 5:] The top pointer is changed in Line~\ref{line:updatetop}.
\end{description}

Let us now explain the steps and mention why these steps need to be performed
in a sequence (not necessarily by the same thread). To insert any node in
Line~\ref{line:pushcas} it is necessary to perform steps 0, 1 and 2 because to
reach Line~\ref{line:pushcas}, it is necessary to satisfy the condition of the
{\em if} statement in Line~\ref{line:ifdone}.  At this point, the $next$
pointer has been updated, and the $top$ pointer has not been updated.  It is
not possible to insert any other request till the $next$ pointer does not
become $null$. This is only possible when we update the $top$ pointer. To
update the $top$ pointer some thread -- either the thread that is doing the
$push$ operation, or some other thread --  needs to successfully perform the
$updateTop$ operation. This can only be achieved if a thread executes Lines
~\ref{line:nextnonnull} to ~\ref{line:updatetop}, or in other words, performs steps 3, 4,
and 5. Before the $top$ pointer is updated in Line~\ref{line:updatetop},
another concurrent $push$ request cannot be successful because two $push$
requests cannot simultaneously perform a successful CAS operation in step 2.  

Now we have proved that if any $push$ operation has successfully completed step
2 (performed the CAS on the $next$ pointer), then till steps 3, 4, and 5 are
performed no other $push$ request can be successful. This means that between
any two successful $push$ operations, the status of the request needs to be
changed (step 4).  Let us now assume that two $push$ requests for the same node
are successful. Let the request that performs step 2 first be $R_i$, and let
the other request be $R_j$. We denote this fact as: $R_i \prec R_j$. 

$R_j$ must read a different value of the stack top in step 0. Otherwise, it
will find $last.next$ to be non-null and it will not proceed beyond step 0.
Subsequently, it will try to read the status of the request in step 1.  Note
that this step is preceded by step 4 of $R_i$. Thus $R_j$ will read the status
to \DONE, and it will not be able to proceed to step 2. Consequently, $R_j$
will not be able to do the $push$ operation done by $R_i$ once again. 

Hence, the lemma stands proved.
\end{proof}

\begin{lemma}
A $push$ request always adds an entry to the top of the stack in Line~\ref{line:pushcas} in
a bounded number of steps. 
\label{lemm:atleast}
\end{lemma}

\begin{proof}
Let us assume that a $push$ request, $R_i$ never gets fulfilled. This can
happen because it either fails the $if$ conditions in Lines~\ref{line:notlast} and \ref{line:nextnull}, 
or the CAS operation in Line~\ref{line:pushcas}. This can only happen if some other $push$
request makes progress. Now, let us assume that $R_i$ has the least phase
number out of all the $push$ requests that remain unfulfilled for an unbounded
amount of time. 

Since $R_i$ has the least phase number out of all the unfulfilled requests, all
other request with a lower phase number must have gotten fulfilled. This means
that there is a point of time at which $R_i$ has the least phase in the $announce$
array. At this point all the threads must be helping $R_i$ to complete its request.
One of the threads needs to perform a successful CAS in Line~\ref{line:pushcas}. Either that thread
or some other thread can update the $top$ pointer. In this manner, request $R_i$ will
get satisfied. 

Hence, it is not possible to have a request, $R_i$ that waits for an infinite amount
of time to get fulfilled. 
\end{proof}

\begin{theorem}
A $push$ request adds an entry only once to the stack in a bounded
number of steps. Futhermore, if we
just consider the $next$ pointers,  the stack 
is always a linked list without duplicate nodes. 
Thus, the $push$ operation is lock-free.
\end{theorem}

\begin{proof}
Lemma~\ref{lemm:atmost} and Lemma~\ref{lemm:atleast} prove that an entry is added
only once (in a bounded number of steps).
Secondly, we are allowed to modify the $next$ pointer of a node only once.
It cannot only point to another node. Each node points to another node that is pushed
to the stack after it because steps 2-5 are executed in sequence. Thus the stack
at all times has a structure similar to a linked list. The end of the linked list
is the stack top. The $pop$ method does not touch the $next$ pointer; hence, this
property is maintained.
\end{proof}

\begin{lemma}
Every $pop$ operation pops just one element or returns an EmptyStackException.
\label{lemm:poponce}
\end{lemma}

\begin{proof}
We start at the stack top ($mytop \leftarrow top.get()$), and proceed towards the
bottom of the stack. If we get any unmarked node, then we mark it. After marking a node
the $pop$ operation is over. Note that there is no helping in the case of a $pop$ 
operation. Hence, other threads do not work on behalf of a thread. As a result only
one node is marked (or popped). If we are not able to mark a node then we return
an EmtpyStackException. 
\end{proof}

\section{The $clean$ Method}


Let us consider the $clean$ method first. It is called by the $tryCleanUp$ method in Line~\ref{line:callclean}.
The aim of the $clean$ method is to clean a set of $W$ contiguous entries in the list (indexed by 
the $prev$ pointers). Let us start out by defining some terminology. Let us define a range of $W$ 
contiguous entries, which has
four distinguished nodes (see Figure~\ref{fig:rangewa}).

\begin{figure}[!htb]
\begin{center}
\includegraphics[width=1.0\columnwidth]{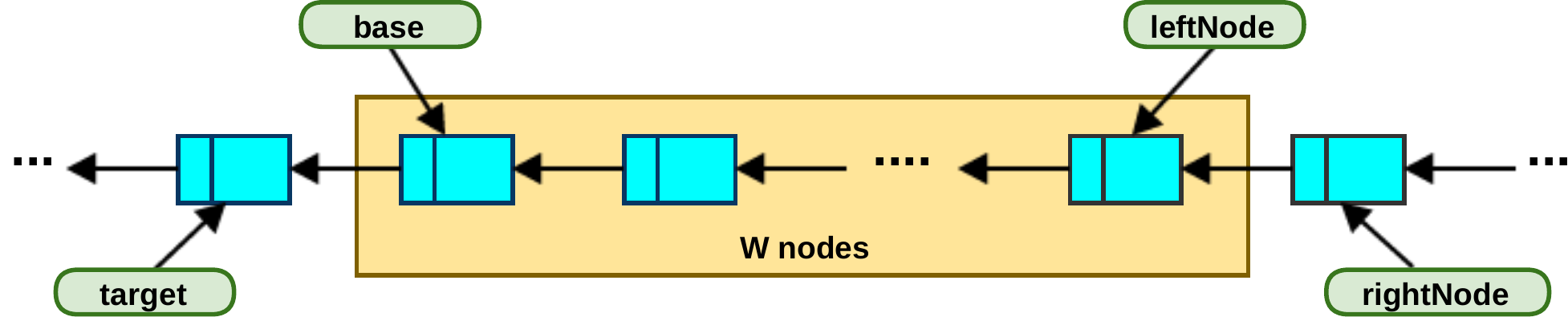}
\caption{A range of $W$ entries \label{fig:rangewa} }
\end{center}
\end{figure}

A range starts with a node termed the $base$, whose index is a multiple of $W$. Let us now define
$target$ as $base.prev$. The node at the end of a range is $leftNode$. Its index is equal to 
$base.index + W - 1$. Let us now define a node $rightNode$ such that $rightNode.prev = leftNode$. 
Note that for a given range, the $base$ and $leftNode$ nodes are fixed, whereas the $target$ and
$rightNode$ nodes keep changing. $rightNode$ is the base of another range, and its index is
a multiple of $W$. 

The $push$ and the $pop$ methods call the function $tryCleanUp$. The $push$ method calls it when
it pushes a node whose index is a multiple of $W$. This is a valid $rightNode$. It walks back 
and increments the counter of the $base$ node of the previous range. We ensure that only one thread
(out of all the helpers) does this in Line~\ref{line:onethread}. Similarly, in the $pop$ function,
whenever we mark a node, we call the $tryCleanUp$ function. Since the $pop$ function does not
have any helpers, only one thread per node calls the $tryCleanUp$ function. Now, inside
the $tryCleanUp$ function, we increment the counter of the $base$ node. Once, a thread increments it to
 $W+1$,
it invokes the $clean$ function. Since only one thread will increment the counter to $W+1$, only one thread
will invoke the $clean$ function for a range.

\begin{algorithm}
\small
\SetAlgoLined
\textbf{class DeleteRequest}{}\\
		\hspace{5mm}$long$ $phase$\\
		\hspace{5mm}$int$ $threadId$\\
		\hspace{5mm}$AtomicBoolean$ $pending$\\
        \hspace{5mm}$Node$ $node$ \\	

\vskip 2mm
AtomicReferenceArray$<$DeleteRequest$>$ allDeleteRequests

\caption{$DeleteRequest$ } \label{alg:deleteclasses}
\end{algorithm}

\begin{algorithm}
\small
\SetAlgoLined
\textbf{clean}($tid$, $node$){}\\
$phase$ $\leftarrow$ $deletePhase.getAndIncrement$()\\
$request$ $\leftarrow$  $new$ DeleteRequest($phase$, $tid$, $true$, $node$) \\
$allDeleteRequests$[$tid$] $\leftarrow$  $request$ \label{line:alldel}\\
$helpDelete$($request$)\\

\textbf{helpDelete}($request$){}\\
			($minTid$, $minReq$)  $\leftarrow$ $min_{req.phase}$ \{ $i$, $req$ $\mid$ $0 \le i < N$, $req$ = $allDeleteRequests[i]$,
$req.pending$ = \textbf{true} \} \\
			\If{($minReq == null$) $\mid\mid$ ($minReq.phase$ $>$ $request.phase$)}{
				break
			}
			$uniqueDelete$($minReq$) \\
			\If{$minReq \ne request$} {
				$uniqueDelete(request)$
			}
	
\textbf{uniqueDelete}($request$){}\\
		\While{$request.pending$}
		{
			$currRequest$ $\leftarrow$ $uniqueRequest.get$()\\
			\If{!$currRequest.pending$ \label{line:deluniquepending} }
			{
				\If{$request.pending$}
				{
					$stat$ $\leftarrow$ ($request \ne currRequest$) ? $uniqueRequest$.$compareAndSet$ ($currRequest$,
$request$) :
\textbf{true} \label{line:cleanCAS}\\
					$helpFinishDelete$()\\
					\If{$stat$} {
						$return$
					}
				}
			}
		\Else{
				$helpFinishDelete$()
			}
		}
\caption{$clean$, $helpDelete$, and $uniqueDelete$ methods} \label{alg:clean}
\end{algorithm}
	
\normalsize
Let us now consider the $clean$, $helpDelete$, and $uniqueDelete$ functions. Their functionality at a high level
is very similar to the $push$ and $help$ methods.  Here, we first create a $DeleteRequest$
that has four fields: $phase$ (similar to phase in $PushOp$), $threadId$, $pending$ (whether the delete
has been finished or not), and  the value of the $base$ node. Akin to the $push$ function,
we add the newly created $DeleteRequest$ to a global array of $DeleteRequest$s in Line~\ref{line:alldel}.
Subsequently, we call the $helpDelete$ function. This function finds a pending request with the minimum phase
in the array $allDeleteRequests$, and returns the request as $minReq$. Subsequently, we invoke $uniqueDelete$.

Note that at this stage it is possible for multiple threads to read the same value of the request with the 
minimum phase number. It is also possible for different sets of threads to have found different requests
to have the minimum phase. For example, if a request with phase 2 ($R_2$) got added to the array before the request with
phase 1 ($R_1$), then a set of threads might be trying to perform $uniqueDelete$ on $R_1$, and another set might
be trying to perform $uniqueDelete$ on $R_1$. Our aim in the $uniqueDelete$ function is to ensure that only
one set goes through to the next stage. It takes two arguments: $req$ (request) and $phase$ (phase number).

We adopt a strategy similar to the one adopted in the function $attachNode$. We define a global atomic
variable, $uniqueRequest$. If a delete is not pending (Line~\ref{line:deluniquepending}) 
on $uniqueRequest$, we read its contents, and try to perform a CAS operation on it. We try to atomically
replace its current contents with the argument, $req$. Note that at this stage, only one set of threads
will be successful. Beyond this point, all the threads will be working on the same DeleteRequest. 
They will then move on to call the $helpFinishDelete$ function that will finish the delete request. 
For threads that are not successful in the CAS operation, or threads that find that the current request
contained in $uniqueRequest$ has a delete pending will also call the $helpFinishDelete$ function. This 
is required to ensure wait freedom.

\begin{algorithm}
\small
\SetAlgoLined
\textbf{helpFinishDelete}(){}\\
		$currRequest$ $\leftarrow$ $uniqueRequest.get$() \label{line:currreada} \\
		\If{$!currRequest.pending$}
		{
			$return$
		}
\vskip 2mm
			$endIdx$ $\leftarrow$ $currRequest.node.index + W - 1$	 \label{line:endidxa}\\

			/* Search for the request from the $top$ */ \\
			$rightNode$ $\leftarrow$ $top.get()$ \\
		    $leftNode$ $\leftarrow$ $rightNode.prev$ \\	
			\While{$leftNode.index \ne endIdx$ $\&\&$ $leftNode\ne sentinel$} {
				$rightNode$ $\leftarrow$ $leftNode$ \label{line:righta} \\
				$leftNode$ $\leftarrow$ $leftNode.prev$ \label{line:lefta}
			}
			\If{$leftNode = sentinel$} {
				$return$ /* some other thread deleted the nodes */
			} \label{line:sentinela}
		
\vskip 2mm
		/* Find the target node */ \\
		$target$ $\leftarrow$ $leftNode$  \label{line:targetstarta}\\
		\For{i=0; i $< W$; i++}{
			$target$ $\leftarrow$ $target.prev$	
		} \label{line:targetenda}

\vskip 2mm
		/* Perform the CAS operation and delete the nodes */\\
		$rightNode.prev.compareAndSet(leftNode, target)$ \label{line:rightcasa}

\vskip 2mm
		/* Set the status of the delete request to not pending*/ \\
		$currRequest.pending$ $\leftarrow$ \textbf{false} \label{line:setpendinga}

\caption{The $helpFinishDelete$ method} \label{alg:helpfinishdeletea}
\end{algorithm}
\normalsize
Lastly, let us describe the $helpFinishDelete$ function in Algorithm~\ref{alg:helpfinishdeletea}. 
We first read the current request from the atomic variable, $uniqueRequest$ in Line~\ref{line:currreada}. If
the request is not pending, then some other helper has completed the request, and we can return from the function.
However, if this is not the case, then we need to complete the delete operation. Our aim now is to find the $target$,
$leftNode$, and $rightNode$. We search for these nodes starting from the stack top.

The index of the $leftNode$ is equal to the index of the node in the current request ($currRequest$) + $W-1$. 
$endIdx$ is set to this value in Line~\ref{line:endidxa}.
Subsequently, in Lines~\ref{line:righta}--\ref{line:sentinela}, we start from the top of the stack, and keep traversing the $prev$
pointers till the index of $leftNode$ is equal to $endIdx$. Once, the equality condition is satisfied
Lines~\ref{line:righta} and \ref{line:lefta} give us the pointers to the $rightNode$ and $leftNode$ respectively. If
we are not able to find the $leftNode$, then it means that another helper has successfully deleted the nodes. We can
thus return. 

The next task is to find the $target$. The $target$ is $W$ hops away from the $leftNode$.
Lines~\ref{line:targetstarta}--\ref{line:targetenda} run a loop $W$ times to find the target. Note that we shall never
have any issues with null pointers because $sentinel.prev$ is set to $sentinel$ itself. Once, we have found the
target, we need to perform a CAS operation on the $prev$ pointer of the $rightNode$. We accomplish this in
Line~\ref{line:rightcasa}. If the $prev$ pointer of $rightNode$ is equal to $leftNode$, then we set it to $target$.
This operation removes $W$ entries (from $leftNode$ to $base$) from the list. The last step is to set the 
status of the $pending$ field in the current request ($currRequest$) to false (see Line~\ref{line:setpendinga}).

\end{appendix}

\end{document}